\documentclass[11pt,reqno]{amsart}

\usepackage{amscd,amssymb,amsmath,amsthm}
\usepackage[arrow,matrix]{xy}
\usepackage{graphicx}
\usepackage{color}
\usepackage{cite}
\newtheorem{thm}{Theorem}

\newtheorem{lemma}{Lemma}
\newtheorem{pro}{Proposition}
\newtheorem{rk}{Remark}
\newtheorem{cor}{Corollary}

\numberwithin{equation}{section} 

\begin{document}
\title[Translation-invariant Gibbs measures for the Potts model]{
Description of the translation-invariant splitting Gibbs measures
for the Potts model on a Cayley tree
}

\author{C. K\"ulske, U. A. Rozikov, R. M. Khakimov}

\address{C.\ K\"ulske\\ Fakult\"at f\"ur Mathematik,
Ruhr-University of Bochum, Postfach 102148,\,
44721, Bochum,
Germany}
\email {Christof.Kuelske@ruhr-uni-bochum.de}

\address{U.\ A.\ Rozikov\\ Institute of mathematics,
29, Do'rmon Yo'li str., 100125, Tashkent, Uzbekistan.}
\email {rozikovu@yandex.ru}

\address{R.M. Khakimov\\ Namangan state university, Namangan, Uzbekistan.} \email
{rustam-7102@rambler.ru}

\begin{abstract} For the $q$-state Potts model on a Cayley tree of order $k\geq 2$ it is well-known that
at sufficiently low temperatures there are at least $q+1$ translation-invariant Gibbs measures
which are also tree-indexed Markov chains. Such measures are called  translation-invariant
splitting Gibbs measures (TISGMs).

In this paper we find all TISGMs, and show in particular that at sufficiently low
temperatures their number is  $2^{q}-1$.
We prove that there are $[q/2]$ (where $[a]$ is the integer part of $a$) critical temperatures at which
the number of TISGMs changes
and give the exact number of TISGMs
for each intermediate temperature.
For the binary tree we give explicit formulae for the
critical temperatures and the possible TISGMs.

While we show that these measures are never convex combinations of each other, the question
which of these measures are extremals in the set of all Gibbs measures will be treated in
future work.
\end{abstract}
\maketitle



\section{Introduction}
It is known for the Ising model on a Cayley tree that
the number of translation-invariant 
splitting Gibbs measures
can only be one or three, depending on temperature.
For the $q$-state Potts model on a Cayley tree
it is known that for some temperatures there exist either 1 or {\bf at least}
$q+1$ TISGMs \cite{Ga8}, \cite{Ga9}, corresponding to homogeneous boundary conditions
equal to either one of the $q$ possible spin values or free boundary conditions.
However, a complete result about {\bf all} splitting Gibbs measures is lacking.

In this paper, we shall fully describe the set of TISGMs for the Potts model.
 We recall in this context that all extremal Gibbs measures on trees are necessarily Markov chains
 \cite{Ge}. Hence we found all extremal translation-invariant Gibbs measures for the Potts model. Note also that not all Gibbs measures which are
Markov chains (splitting Gibbs measures, SGMs)
have to be extremal measures (in the set of all Gibbs measures). This is the case
for the state obtained with free boundary conditions in some temperature interval below the transition
temperature for which there is more than one Gibbs measure.

In this paper we consider ferromagnetic Potts model.
The ferromagnetic Potts model is studied among others by N.Ganikhodjaev who proved that for some temperatures it possesses at least  $q+1$ TISGMs and an uncountable number
of  extremal non-translation invariant splitting Gibbs measures \cite{Ga8}, \cite{Ga9}.
For other results related to the Potts model on Cayley trees see
\cite{Ag1}-\cite{Per5}, \cite{Ro9}, \cite{Wa}, \cite{Wu} and the references therein.

Our main result of this paper
is the characterization and counting of TISGMs which is given in Theorem \ref{Theorem4}. Let us outline the proof.
Our analysis is based on a systematic investigation of the tree recursion for boundary fields
(boundary laws) whose fixed points are
characterizing the TISGMs. In this analysis we find {\bf all} fixed points.
We show that these fixed points can be characterized according to the number of their non-zero
components, see Theorem \ref{Theorem2}. Care is needed, since not all of these solutions give rise to different Gibbs measures, and we have to take into account of symmetries in a proper way
when going from the full description of fixed points to the full description of TISGMs
of Theorem \ref{Theorem4}. \bigskip

{\it Cayley tree.}
The Cayley tree $\Gamma^k$
of order $ k\geq 1 $ is an infinite tree, i.e., a graph without
cycles, such that exactly $k+1$ edges originate from each vertex.
Let $\Gamma^k=(V, L)$ where $V$ is the set of vertices and  $L$ the set of edges.
Two vertices $x$ and $y$ are called {\it nearest neighbors} if there exists an
edge $l \in L$ connecting them.
We will use the notation $l=\langle x,y\rangle$.
A collection of nearest neighbor pairs $\langle x,x_1\rangle, \langle x_1,x_2\rangle,...,\langle x_{d-1},y\rangle$ is called a {\it
path} from $x$ to $y$. The distance $d(x,y)$ on the Cayley tree is the number of edges of the shortest path from $x$ to $y$.

For a fixed $x^0\in V$, called the root, we set
\begin{equation*}
W_n=\{x\in V\,| \, d(x,x^0)=n\}, \qquad V_n=\bigcup_{m=0}^n W_m
\end{equation*}
and denote
$$
S(x)=\{y\in W_{n+1} :  d(x,y)=1 \}, \ \ x\in W_n, $$ the set  of {\it direct successors} of $x$.

{\it The model.} We consider the {\it Potts model} on a Cayley tree,
where each spin $\sigma(x)$ takes values in the set
$\Phi:=\{1,2,\dots,q\}$, and spins are assigned to the vertices
of the tree.

The (formal) Hamiltonian of Potts model is
\begin{equation}\label{ph}
H(\sigma)=-J\sum_{\langle x,y\rangle\in L}
\delta_{\sigma(x)\sigma(y)},
\end{equation}
where $J\in R$ is a coupling constant,
$\langle x,y\rangle$ stands for nearest neighbor vertices and $\delta_{ij}$ is the Kronecker
symbol:
$$\delta_{ij}=\left\{\begin{array}{ll}
0, \ \ \mbox{if} \ \ i\ne j\\[2mm]
1, \ \ \mbox{if} \ \ i= j.
\end{array}\right.
$$

{\it Gibbs measure.} Following \cite{Pr} (and subsequent works see \cite{BRZ}, \cite{Ro}), we consider a special class
of Gibbs measures. These measures are called in \cite{Pr}, \cite{Ge} Markov chains and in \cite{Za}, \cite{Za1}
entrance laws. In this paper we call them splitting Gibbs measures (as in \cite{Ro12}).

Define a finite-dimensional distribution of a probability measure $\mu$ in the volume $V_n$ as
\begin{equation}\label{p*}
\mu_n(\sigma_n)=Z_n^{-1}\exp\left\{-\beta H_n(\sigma_n)+\sum_{x\in W_n}{\tilde h}_{\sigma(x),x}\right\},
\end{equation}
where $\beta=1/T$, $T>0$ is the temperature,  $Z_n^{-1}$ is the normalizing factor, $\{{\tilde h}_x=({\tilde h}_{1,x},\dots, {\tilde h}_{q,x})\in R^q, x\in V\}$ is a collection of vectors and
$H_n(\sigma_n)$ is the restriction of the Hamiltonian on $V_n$.

We say that the probability distributions (\ref{p*}) are compatible if for all
$n\geq 1$ and $\sigma_{n-1}\in \Phi^{V_{n-1}}$:
\begin{equation}\label{p**}
\sum_{\omega_n\in \Phi^{W_n}}\mu_n(\sigma_{n-1}\vee \omega_n)=\mu_{n-1}(\sigma_{n-1}).
\end{equation}
Here $\sigma_{n-1}\vee \omega_n$ is the concatenation of the configurations.
In this case, there exists a unique measure $\mu$ on $\Phi^V$ such that,
for all $n$ and $\sigma_n\in \Phi^{V_n}$,
$$\mu(\{\sigma|_{V_n}=\sigma_n\})=\mu_n(\sigma_n).$$

Such a measure is called a {\it splitting Gibbs measure} (SGM) corresponding to the Hamiltonian (\ref{ph}) and vector-valued function ${\tilde h}_x, x\in V$.

It is known (see \cite{Ga8}, \cite[p.106]{Ro}) that
the probability distributions
$\mu_n(\sigma_n)$, $n=1,2,\ldots$, in
(\ref{p*}) are compatible for the Potts model iff for any $x\in V\setminus\{x^0\}$
the following equation holds:
\begin{equation}\label{p***}
 h_x=\sum_{y\in S(x)}F(h_y,\theta),
\end{equation}
where $F: h=(h_1, \dots,h_{q-1})\in R^{q-1}\to F(h,\theta)=(F_1,\dots,F_{q-1})\in R^{q-1}$ is defined as
$$F_i=\ln\left({(\theta-1)e^{h_i}+\sum_{j=1}^{q-1}e^{h_j}+1\over \theta+ \sum_{j=1}^{q-1}e^{h_j}}\right),$$
$\theta=\exp(J\beta)$, $S(x)$ is the set of direct successors of $x$ and $h_x=\left(h_{1,x},\dots,h_{q-1,x}\right)$ with
\begin{equation}\label{hh}
h_{i,x}={\tilde h}_{i,x}-{\tilde h}_{q,x}, \ \ i=1,\dots,q-1.
\end{equation}

Hence for any $h=\{h_x,\ \ x\in V\}$
satisfying (\ref{p***}) there exists a unique SGM $\mu$ for the Potts model.

\section{Results}

In this paper we consider SGMs which are translation-invariant, i.e., we assume $h_x=h=(h_1,\dots,h_{q-1})\in R^{q-1}$ for all $x\in V$.

Put
$$T_{cr}={J\over \ln\left(1+{q\over k-1}\right)}.$$

The main result of this paper is the following

\begin{thm} \label{Theorem4}
For the $q$-state ferromagnetic ($J>0$) Potts model on the Cayley tree of order $k\geq 2$ there are critical temperatures $T_{c,m}\equiv T_{c,m}(k,q)$, $m=1,\dots,[q/2]$ such that the following statements hold

\begin{itemize}
\item[1.] $$T_{c,1}>T_{c,2}>\dots>T_{c,[{q\over 2}]-1}>T_{c,[{q\over 2}]}\geq T_{cr};$$

\item[2.]
If $T>T_{c,1}$ then there exists a unique TISGM;

\item[3.]
If $T_{c,m+1}<T<T_{c,m}$ for some $m=1,\dots,[{q\over 2}]-1$ then there are
$1+2\sum_{s=1}^m{q\choose s}$ TISGMs.

\item[4.] If $T_{cr}\ne T<T_{c,[{q\over 2}]}$ then there are $2^{q}-1$ TISGMs.

\item[5.] If $T=T_{cr}$ then
the number of TISGMs is as follows
$$\left\{\begin{array}{ll}
2^{q-1}, \ \ \mbox{if} \ \ q -odd\\[2mm]
2^{q-1}-{q-1\choose q/2}, \ \ \mbox{if} \ \ q -even.
\end{array}\right.;$$

\item[6.] If $T=T_{c,m}$, $m=1,\dots,[{q\over 2}]$, \, ($T_{c,[q/2]}\ne T_{cr}$) then  there are
$1+{q\choose m}+2\sum_{s=1}^{m-1}{q\choose s}$ TISGMs.
\end{itemize}
\end{thm}

\begin{rk}
Theorem \ref{Theorem4} describes all possible TISGMs of the ferromagnetic Potts model.
In our next paper we shall find some regions for the temperature parameter
ensuring that a given TISGM is (non-)extreme in the set of all Gibbs measures.
We note that this is a technically difficult problem, and one can not expect
sharp bounds on temperatures for which extremality and non-extremality holds for
all types of states, as already the literature in the case of the
 free state in the Potts model
 (the so-called reconstruction problem) shows:
For a complete characterization of the temperature intervals
in  which extremality (resp. non-extremality) holds
it is in general neither sufficient to decide on stability of solutions, nor simply to check
the second largest eigenvalue of the corresponding transition matrix,
although both gives valuable partial information.
\end{rk}

The critical temperatures are given when $k=2$ by
\begin{equation}\label{T2}
T_{c,m}\equiv T_{c,m}(2,q)={J\over \ln\left(1+2\sqrt{m(q-m)}\right)}, \ \ m=1,2,\dots [{q\over 2}].
\end{equation}

From Theorem \ref{Theorem4} we get the following
\begin{cor}\label{co1} If $T\leq T_{c,1}$ then there are at least two extreme (not necessarily TI) Gibbs measures for the $q$-state Potts model on the Cayley tree.
\end{cor}

Let ${\mathcal G}_{TI}$ be the set of all TISGMs. We have the following structural
statement, expressing an independence of the states in ${\mathcal G}_{TI}$ of each other.

\begin{thm}\label{tex} All measures  $\mu\in {\mathcal G}_{TI}$
are extremal points of the convex hull of the set ${\mathcal G}_{TI}$.
\end{thm}

The theorem says equivalently:
Any measure $\mu\in {\mathcal G}_{TI}$ can not be non-trivial
mixture of measures from ${\mathcal G}_{TI}\setminus \{\mu\}$.
Again, we note that it could be a non-trivial mixture of Gibbs-measures
of non-translation invariant states.

\section{Proofs}
For $h_x=h$ from equation (\ref{p***}) we get $h=kF(h,\theta)$, i.e.,
\begin{equation}\label{pt}
h_i=k\ln\left({(\theta-1)e^{h_i}+\sum_{j=1}^{q-1}e^{h_j}+1\over \theta+ \sum_{j=1}^{q-1}e^{h_j}}\right),\ \ i=1,\dots,q-1.
\end{equation}
Denoting $z_i=\exp(h_i), i=1,\dots,q-1$, we get from (\ref{pt})
\begin{equation}\label{pt1}
z_i=\left({(\theta-1)z_i+\sum_{j=1}^{q-1}z_j+1\over \theta+ \sum_{j=1}^{q-1}z_j}\right)^k,\ \ i=1,\dots,q-1.
\end{equation}

The following proposition plays the key role in our proofs:

\begin{pro}\label{tti}  \label{Theorem2} For any solution $z=(z_1,\dots,z_{q-1})$ of the system of equations (\ref{pt1}) there exists $M\subset \{1,\dots,q-1\}$ and $z^*>0$ such that
$$z_i=\left\{\begin{array}{ll}
1, \ \ \mbox{if} \ \ i\notin M\\[3mm]
z^*, \ \ \mbox{if} \ \ i\in M.
\end{array}
\right.
$$
\end{pro}
\begin{proof} It is easy to see that $z_i=1$ is a solution of $i$th equation of the system (\ref{pt1}) for each $i=1,2,\dots,q-1$. Thus for a given $M\subset \{1,\dots,q-1\}$ one can take $z_i=1$ for any $i\notin M$. Let $\emptyset\ne M\subset\{1,\dots,q-1\}$, without loss of generality we can take $M=\{1,2, \dots, m\}$, $m\leq q-1$, i.e. $z_i=1$, $i=m+1,\dots, q$. Now we shall
prove that $z_1=z_2=\dots=z_m$.  Denote $\sqrt[k]{z_i}=x_i$, $i=1,\dots,m$. Then from (\ref{pt1}) we have
\begin{equation}\label{r1}
    x_i=\frac{(\theta-1)x_i^k+\sum_{j=1}^mx_j^k+q-m}{\sum_{j=1}^mx_j^k+q-m-1+\theta}, \ \ i=1,\dots,m.
\end{equation}
  By assumption $x_i\neq1, i=1,2,...,m$ from (\ref{r1}) we get
$$\theta-1=\frac{(x_i-1)\left(\sum_{j=1}^mx_j^k+q-m\right)}{x_i^k-x_i}=
\frac{\sum_{j=1}^mx_j^k+q-m}{x_i(x_i^{k-2}+x^{k-3}_i+\dots+1)}, \ \ i=1,\dots,m.$$
From these equations we get
$$x_i(x_i^{k-2}+x^{k-3}_i+\dots+1)=x_j(x_j^{k-2}+x^{k-3}_j+\dots+1),$$
 which gives $x_i=x_j$ for any $i,j\in \{1,\dots,m\}$.
\end{proof}
By this proposition we have that any TISGM of the Potts model corresponds to a solution of the following equation
\begin{equation}\label{rm}
z=f_m(z)\equiv \left({(\theta+m-1)z+q-m\over mz+q-m-1+\theta}\right)^k,
\end{equation}
for some $m=1,\dots,q-1$.

\begin{lemma}\label{l1} If $z(m_1)$ is a solution to (\ref{rm}) with $m=m_1$ then $z^{-1}(m_1)$ is a
solution to (\ref{rm}) with $m=q-m_1$.
\end{lemma}
\begin{proof} It is easy to see that $f_{m}(x)=1/f_{q-m}(x^{-1})$.
\end{proof}

Let $M\subset \{1,\dots, q-1\}$, with $|M|=m$. Then we denote the corresponding solution of (\ref{rm})
by $z(M)=\exp(h(M))$. It is clear that $h(M)=h(m)$, i.e. it only depends on the cardinality of $M$.
Put
$${\mathbf 1}_M=(e_1,\dots,e_q), \ \ \mbox{with} \ \ e_i=1 \ \ \mbox{if} \ \ i\in M, \ \ e_i=0 \ \ \mbox{if} \ \ i\notin M.$$

 We denote by $\mu_{h(M)\mathbf 1_M}$ the TISGM corresponding to the solution $h(M)$.
\begin{rk}\label{ns} By formula (\ref{hh}) we have
$${\tilde h}_i(M)=\ln(\tilde{z}_i(M))=\left\{\begin{array}{ll}
h(M)+{\tilde h}_q(M), \ \ \mbox{if} \ \ i\in M\\[2mm]
{\tilde h}_q(M), \ \ \mbox{if} \ \ i\notin M
\end{array}\right.,$$ i.e.
$${\tilde h}(M){\mathbf 1}_M=h(M){\mathbf 1}_M+{\tilde h}_q(M){\mathbf 1}_{\{1,\dots,q\}}.$$
 Hence for a given $M$, $|M|=m$ and a solution $h(M)$ the number of vectors ${\tilde h}(M){\mathbf 1}_M$ is equal to ${q\choose m}$.
\end{rk}

The following proposition is useful.

\begin{pro}\label{tp} For any finite $\Lambda\subset V$ and any $\sigma_\Lambda\in \{1,\dots,q\}^\Lambda$ we have
\begin{equation}\label{mu}
\mu_{h(M){\mathbf 1}_M}(\sigma_\Lambda)=\mu_{h(M^c){\mathbf 1}_{M^c}}(\sigma_\Lambda),
\end{equation}
where $M^c=\{1,\dots,q\}\setminus M$ and $h(M^c)=-h(M)$.
\end{pro}
\begin{proof} Let ${\mathbb I}$ be the indicator function. Then we have ${\mathbb I}(x\in M)+{\mathbb I}(x\in M^c)=1$.
By Lemma \ref{l1} we get
$$1/z(M)=z(M^c), \ \ \mbox{i.e.}\ \ -h(M)=h(M^c).$$
Using these formulas and Remark \ref{ns} for $\partial \Lambda=\{x\in V\setminus \Lambda: \exists y\in \Lambda,  \langle x,y\rangle\}$ we get
$$  \mu_{h(M){\mathbf 1}_M}(\sigma_\Lambda)={\exp\left(-\beta H(\sigma_\Lambda)\right)(\tilde{z}_q(M))^{|\partial \Lambda|}(z(M))^{\sum_{x\in \partial \Lambda}{\mathbb I}(\sigma(x)\in M)}\over \sum_{\varphi_\Lambda}\exp\left(-\beta H(\varphi_\Lambda)\right)(\tilde{z}_q(M))^{|\partial \Lambda|}(z(M))^{\sum_{x\in \partial \Lambda}{\mathbb I}(\varphi(x)\in M)}}$$
$$  ={\exp\left(-\beta H(\sigma_\Lambda)\right)(z(M))^{|\partial \Lambda|-\sum_{x\in \partial \Lambda}{\mathbb I}(\sigma(x)\in M^c)}\over \sum_{\varphi_\Lambda}\exp\left(-\beta H(\varphi_\Lambda)\right)(z(M))^{|\partial \Lambda|-\sum_{x\in \partial \Lambda}{\mathbb I}(\varphi(x)\in M^c)}}$$
$$  ={\exp\left(-\beta H(\sigma_\Lambda)\right)(1/z(M))^{\sum_{x\in \partial \Lambda}{\mathbb I}(\sigma(x)\in M^c)}\over \sum_{\varphi_\Lambda}\exp\left(-\beta H(\varphi_\Lambda)\right)(1/z(M))^{\sum_{x\in \partial \Lambda}{\mathbb I}(\varphi(x)\in M^c)}}$$
$$  ={\exp\left(-\beta H(\sigma_\Lambda)\right)(\tilde{z}_q(M^c))^{|\partial \Lambda|}(z(M^c))^{\sum_{x\in \partial \Lambda}{\mathbb I}(\sigma(x)\in M^c)}\over \sum_{\varphi_\Lambda}\exp\left(-\beta H(\varphi_\Lambda)\right)
(\tilde{z}_q(M^c))^{|\partial \Lambda|}(z(M^c))^{\sum_{x\in \partial \Lambda}{\mathbb I}(\varphi(x)\in M^c)}}=\mu_{h(M^c){\mathbf 1}_{M^c}}(\sigma_\Lambda).$$
\end{proof}
The following is a corollary of Propositions \ref{tti} and \ref{tp}.
\begin{cor}\label{c1} Each TISGM corresponds to a solution of (\ref{rm}) with some $m\leq [q/2]$.
Moreover, for a given $m\leq [q/2]$,
a fixed solution to (\ref{rm}) generates ${q\choose m}$ vectors ${\tilde h}$ giving ${q\choose m}$ TISGMs.
\end{cor}

Now to find {\it explicit} solutions of (\ref{rm}) and to obtain {\it explicit} formulas of critical temperatures we assume $k=2$. Then denoting $x=\sqrt{z}$ from (\ref{rm}) we get
\begin{equation}\label{qe}
mx^3-(\theta+m-1)x^2+(\theta+q-m-1)x-q+m=0.
\end{equation}
Note that $x=1$ is a solution to (\ref{qe}). Dividing the left hand side of this equation by $x-1$ we obtain
\begin{equation}\label{qe1}
mx^2-(\theta-1)x+q-m=0.
\end{equation}
This equation has solutions
\begin{equation}\label{s}
x_1(m)={\theta-1-\sqrt{(\theta-1)^2-4m(q-m)}\over 2m}, \ \ x_2(m)={\theta-1+\sqrt{(\theta-1)^2-4m(q-m)}\over 2m},
\end{equation}
if
\begin{equation}\label{tm}
\theta\geq \theta_m=1+2\sqrt{m(q-m)}, \ \ m=1,\dots,q-1.
\end{equation}
From this formula of $\theta_m$ we get (\ref{T2}).

\begin{rk}\label{rr} We note that if $x>0$ is a common solution to equations (\ref{qe1}) for different values of $m$, say $m_1$ and $m_2$, then $x=1$. Indeed, for this $x$ we should have
$$
m_1x^2-(\theta-1)x+q-m_1=0, \ \ m_2x^2-(\theta-1)x+q-m_2=0.$$
Subtracting from the first equation the second one we obtain
$$(m_1-m_2)(x^2-1)=0.$$ Hence, in case $m_1\ne m_2$ we have only $x=1$.

Moreover, $x_1(m)\ne 1$ and $x_2(m)\ne 1$ if $\theta\ne q+1$. If $\theta=q+1$ then
\begin{equation}\label{q1}
x_1(m)=\left\{\begin{array}{ll}
1, \ \ \mbox{if} \ \ q\geq 2m\\[3mm]
{q\over m}-1, \ \ \mbox{if} \ \ q< 2m;
\end{array}\right.
 \ \ x_2(m)=\left\{\begin{array}{ll}
1, \ \ \mbox{if} \ \ q\leq 2m\\[3mm]
{q\over m}-1, \ \ \mbox{if} \ \ q>2m.
\end{array}\right.
\end{equation}
In particular, $x_1(m)=x_2(m)=1$ if $q=2m$.
\end{rk}

It is easy to see that
\begin{equation}\label{st}
\theta_m=\theta_{q-m} \ \ \mbox{and} \ \ \theta_{1}<\theta_2<\dots<\theta_{[{q\over 2}]-1}<\theta_{[{q\over 2}]}\leq q+1.
\end{equation}

Thus we have the following

\begin{pro}\label{pw} Let $k=2$, $J>0$.
\begin{itemize}
\item[1.]
If $\theta<\theta_1$ then the system of equations (\ref{pt}) has
a unique solution $h_0=(0,0,\dots,0)$;

\item[2.]
If $\theta_{m}<\theta<\theta_{m+1}$ for some $m=1,\dots,[{q\over 2}]-1$ then the system of equations (\ref{pt}) has
solutions
$$h_0=(0,0,\dots,0), \ \ h_{1i}(s), \ \ h_{2i}(s), \ \ i=1,\dots, {q-1\choose s}, $$
$$h'_{1i}(q-s), \ \ h'_{2i}(q-s), \ \ i=1,\dots, {q-1\choose q-s}, \ \ s=1,2,\dots,m,$$
where $h_{ji}(s)$, (resp. $h'_{ji}(q-s)$)\, $j=1,2$ is a vector with $s$ (resp. $q-s$) coordinates equal to $2\ln x_j(s)$ (resp. $2\ln x_j(q-s)$) and the remaining $q-s-1$ (resp. $s-1$) coordinates equal to 0. The number of such solutions is equal to
$$1+2\sum_{s=1}^m{q\choose s};$$

\item[3.] If $\theta_{[{q\over 2}]}<\theta\ne q+1$ then there are $2^q-1$ solutions to (\ref{pt});

\item[4] If $\theta=q+1$ then the
number of solutions is as follows
$$\left\{\begin{array}{ll}
2^{q-1}, \ \ \mbox{if} \ \ q  \ \ \mbox{is odd}\\[2mm]
2^{q-1}-{q-1\choose q/2}, \ \ \mbox{if} \ \ q \ \ \mbox{is even};
\end{array}\right.$$

\item[5.] If $\theta=\theta_m$, $m=1,\dots,[{q\over 2}]$, \,($\theta_{[{q\over 2}]}\ne q+1$) then  $h_{1i}(m)=h_{2i}(m)$. The number of solutions is equal to
$$1+{q\choose m}+2\sum_{s=1}^{m-1}{q\choose s}.$$
\end{itemize}
\end{pro}
\begin{proof} 1. Straightforward.

2. Fix $m\in \{1,\dots,q-1\}$. Then for each $M\subset \{1,\dots,q-1\}$ with
cardinality $q-m-1$, we can have two solutions (\ref{s}). Note that $m$ is the number of different from 0 coordinates of a solution to the system (\ref{pt}). By (\ref{s})-(\ref{st}) the number of such solutions (vectors) is equal to two times  ${q-1\choose m}+{q-1\choose q-m}={q\choose m}$.

3. If $\theta>\theta_{[q/2]}$ and $\theta\ne q+1$ then the number of solutions depending on the parity of $q$ is as follows:

\begin{equation}\label{qqq}\left\{\begin{array}{ll}
1+2\sum_{s=1}^{[q/2]}{q\choose s}, \ \ \mbox{if} \ \ q \ \ \mbox{is odd}\\[3mm]
1+{q\choose q/2}+2\sum_{s=1}^{q/2-1}{q\choose s}, \ \ \mbox{if} \ \ q \ \ \mbox{is even}
\end{array}\right.=2^q-1,\end{equation}
where we used the following formulas
$${q\choose s}={q\choose q-s}, \ \ \ \sum_{s=1}^{q}{q\choose s}=2^{q}-1.
$$

4. For $\theta=q+1$ we use (\ref{q1}), Remark \ref{rr}, (\ref{st}) and (\ref{qqq}) to get the following formulas for the number of solutions:

{\it $q$ is odd:}
$$1+\sum_{s=1}^{[q/2]}{q\choose s}=2^{q-1};$$

{\it $q$ is even:}
$$1+\sum_{s=1}^{q/2-1}{q\choose s}=2^{q-1}-{q-1\choose q/2},$$
in the last formula we used ${1\over 2}{q\choose q/2}={q-1\choose q/2}.$

5. On critical point $\theta_m$ the equation (\ref{rm}) has a unique solution.
 So omitting a factor 2 in front of ${q\choose m}$ we get the formula in a similar way as in case 2.
 \end{proof}

\begin{rk} \begin{itemize}
\item[1)] By Proposition \ref{pw} for $k=2$ and $J>0$ we have the {\rm full} description of solutions to the system of equations (\ref{pt}). Consequently, this gives the full description of TISGMs. Moreover, by Remark \ref{ns} and Corollary \ref{c1}, depending on parameter $\theta$ the maximal number of such measures can be $2^{q}-1$. This number is larger than the number $q+1$ of known (see \cite{Ga9}) TISGMs of the Potts model with $q\geq 3$.

\item[2)] We note that the well-known critical value of $\theta$ which implies non-uniqueness of Gibbs measure of the Potts model is $\theta_c={k+q-1\over k-1}$, $k\geq 2$, $q\geq 2$ (see \cite{Ga8}, \cite{Ga9}, \cite[p.115]{Ro}).
For $k=2$ by (\ref{tm}) we have
$$\left\{\begin{array}{ll}
\theta_c> \theta_m, \ \ \mbox{for all} \ \ m\in \{1,\dots,q-1\}\setminus \{q/2\}\\[2mm]
\theta_c=\theta_m, \ \ \mbox{for} \ \ m=q/2.
\end{array}
\right.$$
By the Corollary \ref{co1} we improve the critical value, i.e., for the $q$-state Potts model ($q\geq 3$) on the Cayley tree, the Gibbs measure is not unique starting from $\theta=\theta_1$.
\end{itemize}
\end{rk}

\begin{proof} {\sl Proof of Theorem} \ref{Theorem4}. {\bf Step 1.} First we shall show that for each $m$ the equation (\ref{rm})
has up to three solutions. This equation can be rewritten in the following form
\begin{equation}\label{d1}
(x-1)\varphi_m(x,\theta)=0,
\end{equation}
where
$$ \varphi_m(x,\theta)=mx^k-(\theta-1)(x^{k-1}+x^{k-2}+\dots+x)+q-m.$$

It is well known (see \cite[p.28]{Pra}) that the number of positive
roots of a polynomial does not exceed the number of sign
changes of its coefficients. Using this for $\varphi_m$ we see
that the equation
\begin{equation}\label{oe}
\varphi_m(x,\theta)=0
\end{equation} has up to two positive solutions.

{\bf Step 2.} In this step we shall check that Remark \ref{rr} is also true for $k\geq 3$, i.e. if $x>0$ is a common
solution to equation (\ref{oe}) for different values of $m$ then $x=1$.
Indeed, if $x>0$ is a solution to (\ref{oe}) for $m=m_1$ and $m=m_2$ with $m_1\ne m_2$. Then subtracting one equation from the second one we get
$(m_1-m_2)(x^k-1)=0$, i.e $x=1$.

Moreover, it is easy to see that $x=1$ is a solution to (\ref{oe}) iff $\theta=\theta_c={k+q-1\over k-1}$. As mentioned above the equation has up to two solutions, assume now one of the solutions is $x=1$ then dividing
the polynomial $\varphi_m(x,\theta_c)$ by $x-1$ we get
$$\psi_m(x)={\varphi_m(x,\theta_c)\over x-1}=m\sum_{i=0}^{k-1}x^i-{q\over k-1}\sum_{i=0}^{k-2}(k-i-1)x^i.$$
Thus the second solution of (\ref{oe}) will be 1 if $\psi_m(1)=0$ which holds iff $q=2m$.

{\bf Step 3.} Now we shall prove that for each $m$ there exists a unique critical temperature $T_{c,m}$.
A critical value $\theta_m$ of the parameter $\theta$ is the value
for which the equation $\varphi_m(x,\theta)=0$ has a unique solution $x$.
It is easy to see that this critical value satisfies the following system
\begin{equation}\label{vs}
\left\{\begin{array}{ll}
\varphi_m(x, \theta_m)=0\\[2mm]
\varphi'_m(x, \theta_m)=mkx^{k-1}-(\theta_m-1)((k-1)x^{k-2}+(k-2)x^{k-3}+\dots+1)=0.
\end{array}
\right.
\end{equation}
From the first equation of the system (\ref{vs})
we get
\begin{equation}\label{te}
\theta_m-1={mx^k+q-m\over x^{k-1}+x^{k-2}+\dots+x}.
\end{equation}
Substituting this to the second equation of (\ref{vs}) we get
\begin{equation}\label{se}
m\sum_{i=1}^{k-1}ix^{2k-i-1}-(q-m)\sum_{i=1}^{k-1}ix^{i-1}=0.
\end{equation}

By the above-mentioned property of polynomials we see that the equation (\ref{se})  has a unique positive solution, say $x=x_*=x_*(m)$. Then for each $m$ the unique critical value $\theta_m$ is given by (\ref{te}), i.e.,
\begin{equation}\label{te1}
\theta_m=1+{mx_*^k+q-m\over x_*^{k-1}+x_*^{k-2}+\dots+x_*}.
\end{equation}
The critical temperature is $T_{c,m}=J\ln^{-1}(\theta_m)$. We have equality $T_{c,m}=T_{c,q-m}$ which follows from the following simple equality
$$\varphi_m(x,\theta)=x^k\varphi_{q-m}(x^{-1},\theta).$$
Hence it remains critical values for $m=1,\dots,[q/2]$.

{\bf Step 4.} Now we shall show that the critical
temperature $T_{c,m}$ is a decreasing function of $m=1,\dots,[q/2]$, so they are distinct.

In order to show that $x_*$ is a decreasing function of $m$, we find $m$ from (\ref{se}) as a function of $x_*$, i.e.
$$m=\xi(x_*), \ \ \mbox{with} \ \ \xi(x)=q\cdot {\sum_{i=1}^{k-1}ix^{i-1}\over
  \sum_{i=1}^{k-1}i\left(x^{i-1}+x^{2k-i-1}\right)}.$$
We have
$$\xi'(x)=-q\cdot {\sum_{i=1}^{k-1}\sum_{j=1}^{k-1}ij(2k-i-j)x^{2k+i-j-3}\over
  \left(\sum_{i=1}^{k-1}i\left(x^{i-1}+x^{2k-i-1}\right)\right)^2}<0.$$

Hence $m$ is a decreasing function of $x_*$. Consequently, $x_*=\xi^{-1}(m)$ is a decreasing function of $m$.

Now let us write $\theta_m$ as function of $x_*$, i.e. substitute $m=\xi(x_*)$ in (\ref{te1}):
\begin{equation}\label{te2}
\theta_m=\eta(x_*)=1+{\xi(x_*)(x_*^k-1)+q\over \sum_{i=1}^{k-1}x_*^i}.
\end{equation}
We have
\begin{equation}\label{eta}
\eta'(x)={\xi'(x)(x^k-1)\sum_{i=1}^{k-1}x^i+\xi(x)
\sum_{i=1}^{k-1}\left((k-i)x^{k+i-1}+ix^{i-1}\right)-q\sum_{i=1}^{k-1}ix^{i-1}\over
\left(\sum_{i=1}^{k-1}x^i\right)^2}.
\end{equation}
Substituting the formulas for $\xi(x)$ and $\xi'(x)$ in (\ref{eta}) we get
$$
\eta'(x)=-\,{q\over
\left(\sum_{i=1}^{k-1}x^i\right)^2 \left(\sum_{i=1}^{k-1}i\left(x^{i-1}+x^{2k-i-1}\right)\right)^2}\times$$
$$\left[   (x^k-1)\left(\sum_{i=1}^{k-1}x^i\right)
\left(\sum_{i=1}^{k-1}\sum_{j=1}^{k-1}ij(2k-i-j)x^{2k+i-j-3}\right)\right.$$
\begin{equation}\label{eta1}
\left.+\left(\sum_{i=1}^{k-1}ix^{i-1}\right)
\left(\sum_{i=1}^{k-1}i\left(x^{i-1}+x^{2k-i-1}\right)\right)
\left(\sum_{i=1}^{k-1}\left(ix^{2k-i-1}-(k-i)x^{k+i-1}\right)\right)\right].
\end{equation}
We have
$$\sum_{i=1}^{k-1}\left(ix^{2k-i-1}-(k-i)x^{k+i-1}\right)=
\sum_{i=1}^{k-1}ix^{2k-i-1}-\sum_{i'=1}^{k-1}i'x^{2k-i'-1}=0,$$
where $i'=k-i$.
Hence from (\ref{eta1}) we obtain
\begin{equation}\label{eta2}
\eta'(x)=-q\cdot {(x^k-1)\left(\sum_{i=1}^{k-1}\sum_{j=1}^{k-1}ij(2k-i-j)x^{2k+i-j-3}\right)\over
\left(\sum_{i=1}^{k-1}x^i\right) \left(\sum_{i=1}^{k-1}i\left(x^{i-1}+x^{2k-i-1}\right)\right)^2}.
\end{equation}

Note that $\xi(1)=q/2$,
thus for $m\in \{1,\dots, [q/2]\}$ we have $x_*(m)>1$. Thus for $x>1$ by (\ref{eta2}) we have
$\eta'(x)<0$. Consequently, $\theta_m$ is a decreasing function of $x_*$. Since $x_*$ is a decreasing function of $m$ we conclude that $\theta_m$ is an increasing function of $m\in \{1,\dots,[q/2]\}$. Moreover, $$\max_{x\geq 1}\eta(x)=\eta(1)=\theta_c=1+{q\over k-1}.$$
Hence using formula $\theta_m=\exp(J/T_{c,m})$ we get $T_{c,1}>T_{c,2}>\dots>T_{c,[{q\over 2}]}\geq T_{cr}.$

{\bf Step 5.} In this last step we calculate the number of TISGMs. In above steps we showed that
if $T>T_{c,m}$ then the equation (\ref{rm}) has the unique solution $x=1$, if $T=T_{c,m}$ then this equation has two solutions $x=1$ and $x=x'$. If $T<T_{c,m}$ then the equation (\ref{rm}) has three solutions $x=1$, $x=x'$ and $x=x''$. Consequently, the number of solutions of the equation (\ref{rm}) is at most three, independently of $k\geq 2$.
Thus the number of solutions is the same as in the case of $k=2$ mentioned in Proposition \ref{pw}, hence using Corollary \ref{c1} and Remark \ref{ns} we get numbers of TISGMs.
This completes the proof.
\end{proof}

\begin{proof} {\sl Proof of Corollary} \ref{co1}. We know that if $T>T_{c,1}$ then there is unique solution $h=(0,\dots,0)$. For any $T>T_{cr}$ the Gibbs measure $\mu_0$ corresponding to $h=(0,\dots,0)$ is extreme (see \cite[Theorem 5.6]{Ro}). For $T\leq T_{c,1}$ we have measure $\mu_0$ and at least $q$ new measures mentioned in Theorem \ref{Theorem4}. If we assume that all the new measures are not extreme then there remains only one known extreme measure $\mu_0$. But in this case the non-extreme measures can not be decomposed only into the unique measure $\mu_0$. Consequently, at least one of the new measures must be extreme or decomposable into other extreme measures.
\end{proof}

\begin{proof} {\sl Proof of Theorem} \ref{tex}. Let $\alpha_i\geq 0$ be such that $\sum_{i=1}^{|{\mathcal G}_{TI}|-1}\alpha_i=1$, and assume that
\begin{equation}\label{mg}
\mu=\sum_{i=1}^{|{\mathcal G}_{TI}|-1}\alpha_i\mu_i.
\end{equation}
Let $z^{(i)}$ be the solution of the equation (\ref{rm}) which corresponds to $\mu_i$.
Since this number does not depend on $x\in V$ we see that our measures satisfy conditions of Corollary (12.18) of \cite{Ge}, by this corollary we have that the RHS of the (\ref{mg}) can not be SGM (or in language of \cite{Ge} a Markov chain).
But in the LHS we have a SGM $\mu$. This is a contradiction.
\end{proof}

\section*{ Acknowledgements}

U.A. Rozikov thanks the  DFG
Sonderforschungsbereich SFB $|$ TR12-Symmetries and Universality in Mesoscopic Systems
and the Ruhr-University Bochum (Germany)
for financial support and hospitality. He also thanks IMU-CDC for a travel support.
We thank both referees for a number
of suggestions which have improved the paper.

\end{document}